\documentclass[aps,pra,twocolumn,10pt,nofootinbib]{revtex4-1}
\usepackage{amsfonts,amssymb,amsmath}            
\usepackage{lmodern}
\usepackage[]{graphics,graphicx}            
\usepackage{amsthm}
\usepackage{bbold,enumerate,mathtools}
    \usepackage{hyperref}

\newtheorem{lemma}{Lemma}

\newcommand{\ket}[1]{\left | #1 \right\rangle}
\newcommand{\bra}[1]{\left \langle #1 \right |}
\newcommand{\half}{\mbox{$\textstyle \frac{1}{2}$}}

\newcommand{\proj}[1]{\ket{#1}\bra{#1}}
\newcommand{\identity}{\mathbb{1}}
\renewcommand{\epsilon}{\varepsilon}
\begin{document}

\title{Perfect Coding for Dephased Quantum State Transfer}
\date{\today}

\author{Alastair \surname{Kay}}
\affiliation{Royal Holloway University of London, Egham, Surrey, TW20 0EX, UK}
\email{alastair.kay@rhul.ac.uk}
\begin{abstract}
We develop a family of perfect quantum error correcting codes that correct for phase errors that arise on any qubit, at any time, during a perfect state transfer experiment. These ensure that we find the optimal operating regime for corrected state transfer. For a specific class of system, we further show that while dephasing noise can be corrected, depolarising noise cannot.
\end{abstract}
\maketitle

Near-future realisations of quantum technologies, from simple state-generation tasks \cite{kay2017,kay2017a,kay2017c} to analogue quantum simulation \cite{sarovar2017}, will be small scale and will rely on intrinsic properties of the system to facilitate as much of their functionality as possible, rather than trying to impose error-prone operations to strong-arm the system into performing tasks. For example, existing quantum simulators such as \cite{bernien2017}, while scaling to larger numbers of qubits, are not universal computational devices; they are Hamiltonian-based systems with some ability to tune their parameters. Other systems similarly sacrifice computational universality to achieve their ends, from quantum key distribution systems \cite{korzh2015} to the DWave quantum computer \cite{king2017}, being tuned to do precisely what they need to. with minimal controls. The difficulty is that Hamiltonian evolution is not obviously compatible with future scaling ambitions. Error correction is a particular challenge given that noise, whose initial action may be well localised, rapidly evolves into potentially destructive correlated errors across an entire device. Can we guard against these effects? The question is especially pertinent when we recognise that interest quantum computers was largely spurred on by the development of a theory of error correction \cite{calderbank1996,steane1996a}. It had previously been suggested that the continuum of possible errors would make quantum error correction vastly more challenging, if not impossible, compared to the discrete errors of classical theory. Implementation of Hamiltonian-based technologies might be similarly inhibited.

In this paper, we study the error correction of one particular protocol, perfect state transfer \cite{bose2003,christandl2004,kay2010a}, which is well suited to this scenario since the relevant observables will be well localised at particular times. We show that error correction is possible with remarkable efficiency by specifying perfect quantum codes, those which are maximally efficient in that every element in the state space is involved in detecting errors, massively improving the operating regime compared to the preliminary results in \cite{kay2016c}. Perfect quantum state transfer is the process whereby an unknown quantum state is transported perfectly from one node of a network to another simply via the evolution of a time-invariant Hamiltonian, often across a one-dimensional spin chain in order to maximise the transfer distance. By suitably engineering the Hamiltonian, this transfer can in principle be achieved over arbitrary distances. In the real-world, including the recent experimental demonstrations \cite{perez-leija2013,chapman2016}, there are always errors, whether these are the result of manufacturing imperfections or noise. Perhaps more worryingly, while these errors might manifest as local Pauli errors, their time-evolved versions, as observed on output from the state transfer, are far from being well-localised.

There are a number of techniques that can compensate for manufacturing imperfections. One option is to make multiple chains and test which is the best before using that one and discarding the others. Alternatively, we encode a state across multiple spins. Optimal encodings can be found across a set of input spins on a single chain \cite{haselgrove2005}, or multiple chains can be used in parallel \cite{burgarth2005}. We are interested in the case where one uses a single chain, encodes in an encoding region, which should be a small fraction of the total chain length, and decodes in a similar sized decoding region at the opposite end of the chain at the end of the evolution.

The challenge of how to deal with noise during the transfer has largely been neglected, perhaps aside from some toy models that have sufficient symmetries that the decoherence is easily avoided \cite{burgarth2006a}. Some early steps were taken in \cite{marletto2012}, showing how an encoding can help ensure a transferring state ``misses'' being at a particular position on the chain at a particular time, but still needed an identification of when and where errors were likely to occur. Then \cite{kay2016c} recently showed that a certain class of standard error correcting code can be adapted to correct for many of the types of error that do arise -- dephasing noise, manufacturing imperfections and timing errors. The codes were unable to tolerate bit-flip noise, and thus relied on the assumption that noise has a dominant direction (the $Z$ direction), which is often reasonable since the $T_2$ decoherence times, corresponding to dephasing noise respectively, can be the dominant one \cite{ithier2005,vandersypen2001,astafiev2004}.

An important feature of the previous study is whether the system outside the encoding area can be initialised in some specific state, such as the all-zeros state (which should be easier to prepare than any arbitrary state). If so, \cite{kay2016c} constructed codes of just 15 qubits. If not, the smallest construction given was 36 qubits. Continuing to assume the presence of dephasing noise, we improve the analysis, enabling massive efficiency savings. We provide the smallest possible codes, utilising just 13 or 15 qubits respectively. These codes can also be applied to compensate for Hamiltonian perturbations and timing errors. Moreover, by specialising slightly more to a class of state transfer chains and just dephasing noise, we find a family of {\em perfect codes} in Section \ref{sec:perfect}, the simplest of which encodes a single logical qubit into just 7 qubits, and protects against a single dephasing error. Such efficiency savings provide an opportunity for radically changing the operating regime for error correction and perfect state transfer. There is only one previously known perfect quantum code \cite{laflamme1996}, in spite of their potential utility. Certainly their classical counterparts, the Hamming and Golay codes, have become ubiquitous.

Following \cite{kay2016c}, this work relies on there being a dominant noise direction. In Section \ref{sec:impossible}, we show a single bit-flip event occurring at an arbitrary time and position during the state transfer process cannot be corrected, making the assumption unavoidable. 

\section{Setting}

We consider one-dimensional Hamiltonians of the form
$$
H=-\half\sum_{n=1}^NB_nZ_n+\half\sum_{n=1}^{N-1}J_n(X_nX_{n+1}+Y_nY_{n+1}),
$$
where the couplings $J_n$ and magnetic fields $B_n$ are selected such that the quantum state $\ket{1}\ket{0}^{\otimes(N-1)}$ evolves in a fixed time $t_0$ into the state $\ket{0}^{\otimes (N-1)}\ket{1}$ (up to a known phase factor), as this ensures that an arbitrary state transfers from one end of the chain to the other. It also implies perfect mirroring of any arbitrary state on the whole chain \cite{albanese2004,kay2010a}, up to a sequence of controlled-phase gates applied between every pair of qubits. There are many such solutions, with the necessary and sufficient conditions being well understood \cite{kay2010a}, but a particularly favourable choice is given by $B_n=0$ and $J_n=\lambda\sqrt{n(N-n)}$ for any positive value $\lambda$, giving $t_0=\pi/(2\lambda)$ \cite{christandl2004,yung2006,kay2006b}. We will refer to this solution as the `standard' state transfer solution.

Of particular relevance to these calculations is the representation of $H$ in the single excitation subspace,
$$
h_1=\sum_{n=1}^NB_n\proj{n}+\sum_{n=1}^{N-1}J_n(\ket{n}\bra{n+1}+\ket{n+1}\bra{n}).
$$
The Majorana fermions, defined as
$$
c_n=X_n\prod_{m=1}^{n-1}Z_m, \qquad c_{n+N}=Y_n\prod_{m=1}^{n-1}Z_m,
$$
each evolve independently under the action of $H$ as
$$
c_n(t)=\sum_{m=1}^{2N}\bra{m}e^{-2iY\otimes h_1 t}\ket{n}c_m.
$$

If a dephasing event, described by $Z_n$, occurs at some time $t$ this is equivalent to $c_nc_{n+N}$ (we ignore phase factors for clarity), and after the end of the state transfer, this is the same as $c_n(t_0-t)c_{n+N}(t_0-t)$ acting on the final state. Thus, any single dephasing event may be described as pairs of Majorana fermion errors. Similarly, timing errors or Hamiltonian perturbations will manifest as a small number of Majorana fermions. We thus aim to find error correcting codes for these errors, when an initial encoding is restricted to a small block of $M$ qubits at the start of the chain, and decoding removes the same $M$ qubits from the opposite end of the chain a time $t_0$ later. In particular, we aim to construct error correcting codes that are distance 5 for Majorana errors.

\subsection{Stabilizer Codes}

We refer to the standard quantum error correcting codes as ``Pauli codes'' to represent the fact that they can correct a certain distance of Pauli errors. Instead, we need to work with Majorana errors in order to tolerate dephasing noise during a state transfer.

Let $S$ be an $(M-k)\times(2M)$ binary matrix that describes the stabilizers of the $[[M,k,d]]$ code -- each row $i$ specifies a stabilizer $S_i$ via the binary vector $(z,x)$ for $z,x\in\{0,1\}^M$, by composing terms
$$
S_i=\prod_{n=1}^MZ_n^{z_n}X_n^{x_n}.
$$
For $S$ to be a stabilizer, the $S_i$ must mutually commute:
$$
S\cdot\Lambda\cdot S^T\equiv 0\text{ mod }2,\qquad\text{where }\Lambda=\left(\begin{array}{cc} 0 & \identity \\ \identity & 0 \end{array}\right).
$$

In the case of Pauli errors, $S$ reduces to a standard form \cite{gottesman1997} via a combination of row reduction (a product of stabilizers is also a stabilizer) and permutation of qubits. Permutation of qubits is straightforward for Pauli errors because swapping a pair of qubits just exchanges error terms, $X_n\leftrightarrow X_{m}$ and $Z_n\leftrightarrow Z_{m}$. However, permutations do not preserve Majorana errors, meaning that this standard form does not apply for Majorana errors. Searches for good error correcting codes consequently span far too large a space to be practical, and we must therefore look at simplified situations.

\subsection{Initial State}

Although we control the first $M$ qubits of the chain, what about initialisation of the rest of the qubits? Here we categorise three options:
\begin{description}
\item[{\bf Case (i)}] the state of the system can be initialised to some fixed state, such as $\ket{0}^{\otimes(N-M)}$, or any density matrix that commutes with $Z^{\otimes(N-M)}$.
\item[{\bf Case (ii)}] the initial state is arbitrary, but we prepare the system suitably by acting only on the encoding and decoding regions.
\item[{\bf Case (iii)}]  the initial state is entirely arbitrary, and we have no ability or desire to prepare it.
\end{description}
Case (i) is not unreasonable as preparing such a state may be significantly easier than preparing an arbitrary state of the rest of the chain. By setting an appropriate (uniform) magnetic field as part of the Hamiltonian, the all-zero state is the ground state, so some form of cooling should be sufficient preparation. Alternatively, measurement of the magnetic field projects into an eigenstate of fixed excitation number which, again, is sufficient.

However, no matter how strong the magnetic field, there will always be some thermal fluctuations, introducing excitations on the chain. Clearly, then, case (iii) is the safest option provided its overheads are not too large. Operation in regime (iii) requires that all stabilizers and codewords of the code commute with $Z^{\otimes M}$, i.e.\ all stabilizers and logical operators contain an even number of bit flips. So, the action of the pair-wise controlled phase gates during perfect state transfer corresponds to an even number of phase gates on each site, i.e.\ nothing happens, and no entanglement is generated between the two blocks. Since $Z^{\otimes M}$ commutes with all logical operators, and the stabilizers, it must itself be a stabilizer (were it to be a logical operator, it would have to anti-commute with another logical operator in order to generate the algebra of a qubit). This is a feature that is easy to impose when searching for candidate codes.

Case (ii) is intermediate between (i) and (iii), and prepares the initial state of the system for state transfer without having to control anything outside the encoding and decoding regions. Here, we take one of the logical operators to be $Z^{\otimes M}$. All the stabilizers commute with it, but the logical $X$ operator does not. To demonstrate, consider no errors, and no error correction. We first prepare a state $\ket{+}=\ket{0}+\ket{1}$ on the last qubit of the chain, wait the perfect state transfer time, and measure the first qubit in the $X$ basis. Had the rest of the chain been in the +1 eigenstate of $Z^{\otimes(N-1)}$, the $\ket{+}$ state arrives perfectly. Had it been in the opposite state, it arrives as $\ket{-}$. Hence projection in the $X$ basis acts to project the rest of the chain into an eigenstate of $Z^{\otimes(N-1)}$. That means an unknown state can now be placed on the encoding region and sent perfectly, no matter what the initial state of the whole system was. In practice, we would have to use the error correcting itself. Given $Z_L=Z^{\otimes M}$ is a logical operator of the code, we prepare an eigenstate of the other logical operator, $X_L$, on the decoding region, wait the perfect transfer time, and then perform error correction on the encoding region, ultimately projecting onto $X_L$, exactly mirroring the unencoded case. The codes in Sec.\ \ref{sec:perfect} are an example of this. To our knowledge, this usage scenario has not previously been expressed.

\section{Dephasing Codes from Pauli codes} \label{sec:stabs}

If the matrix $\tilde S$ describes the stabilizers of some code, how do we know what distance the code has for a particular set of errors? Let us use the binary matrix $E$ to describe these errors, each of which is a product of Paulis. Each column, of length $2M$ is a different error. The first $M$ bits specify the locations of the $Z$ errors, while the second set convey the locations of the $X$ errors. For example, the $E$ corresponding to the $X$ and $Z$ Pauli errors on the $M$ qubits is just a $2M\times 2M$ identity matrix. Each column of $\tilde S\cdot\Lambda\cdot E$ (all arithmetic is calculated modulo 2) tells us which stabilizers give a violation for that error (the $+1$ values). If the code is distance $d$ then all sets of $(d-1)$ columns are linearly independent.

A convenient way to explore the possible $\tilde S$ when $E$ corresponds to the Majorana errors is to look at possible matrices $S=\tilde S\cdot\Lambda\cdot E$, and then invert the relation. Moreover, we already know a useful set of possible matrices -- the Pauli error correcting codes of distance $d$, because these are matrices such that their columns are $(d-1)$-wise linearly independent by definition. Of course they satisfy extra properties that we may not need ($S\cdot\Lambda\cdot S^T\equiv 0$), and do not automatically give back the commutation $\tilde S\cdot\Lambda\cdot \tilde S^T\equiv 0$.

For Majorana errors, the most natural way of writing down $E$ makes the mapping
$$
X_n\leftrightarrow c_n\qquad Z_n\leftrightarrow c_{N+n},
$$
perhaps up to different labellings. This means that
$$
E=\left(\begin{array}{cc}
J^U & J^U+\identity \\
\identity & \identity
\end{array}
\right)P
$$
where $P$ is a permutation matrix that controls those labellings, and $J^U$ denotes a matrix whose upper triangular elements (not including the diagonal) are all ones, and zero otherwise. In \cite{kay2016c}, we implicitly chose a different mapping between Pauli errors and the Majorana fermions: $c_n\leftrightarrow X_n$ and $Z_n=c_nc_{N+n}\leftrightarrow Z_n$, giving
$$
E'=\left(\begin{array}{cc}
\identity & J^U \\
0 & \identity
\end{array}\right)P.
$$
The different choices may be best situated for different instances. For example, if we wish to correct for a generic two-fermion error, the code associated with $E$ would have to be distance 5. However, with the types of error described by $E'$, we will show in Section \ref{sec:css} that it is sufficient to have a distance 5 code for $X$-type errors, but only a distance 3 code for the $Z$ errors in spite of the fact that there can be pairs of $Z$ errors. This permits the design of a smaller code. Nevertheless, there are further simplifications that can be made in certain special cases for which $E$ is optimal (once the required code distance has been updated appropriately) by virtue of constructing a perfect code, see Section \ref{sec:perfect}.


We will now prove that if $S$ is a CSS code, it has the same distance for Majorana errors as it does for Pauli errors. Let us therefore consider a CSS form for $S$ wherein
$$
S=\left(\begin{array}{cc} H_1 & 0 \\ 0 & G_2 \end{array}\right).
$$
$H_1$ represents the parity check matrix for a code $C_1$ of distance $d_1$, while $G_2$ represents the generator matrix for a code $C_2$ that satisfies $C_2\subseteq C_1$ (i.e.\ $H_1\cdot G_2^T\equiv 0$). $G_2$ is also the parity check matrix for the dual to $C_2$, $C_2^\perp$, which has distance $d_2$. We use $G_1$ to denote the generators of $C_1$. By construction, $S$ is distance $d_1$ for $X$ errors and distance $d_2$ for $Z$ errors. Upon calculating the stabilizer violations in the two possible cases of $E$ and $E'$, we have (neglecting the permutation matrices, which simply reorder the columns) respectively
$$
\left(\begin{array}{cc} H_1 & H_1 \\ G_2J^U & G_2(J^U+\identity) \end{array}\right),\qquad \left(\begin{array}{cc} 0 & H_1 \\ G_2 & G_2J^U \end{array}\right).
$$
In either case, the minimum number of columns required to find linear dependence is preserved. Take the $E'$ case. An error $c_n$ gives stabilizer violations described by column $N+n$ in the above matrix. What is the smallest number of other errors that would have to occur in order to make the set of stabilizer violations 0? In order to make the top half 0 (even ignoring the bottom half), you have to introduce a further $d_1-1$ errors of the form $c_m$. Thus, the code is at least distance $d_1$ to the $c_n$ errors. An error $Z_n$ gives a stabilizer violation corresponding to column $n$. It would take a further $d_2-1$ $Z$ errors in order to get a linearly dependent set, so the code is distance $d_2$ to $Z$ errors. Similar arguments can be made for $E$.

We thus conclude that it is sufficient to directly select $\tilde S$ to be a CSS code. There are two different ways in which we might represent our final codes. Once we find a suitable CSS code, we can either use it as the set of stabilizer violations ($S$) and calculate the stabilizers themselves from $\tilde S=S\cdot E^{-1}$, or we set the stabilizers $\tilde S$ to be of the CSS form and we calculate the stabilizer violations to be $S=\tilde S\cdot E$. We will use the latter.

We can now simply search through possible CSS codes. For case (iii), we need the stabilizer of the Majorana code to contain $Z^{\otimes M}$, while case (ii) requires commutation with $Z^{\otimes M}$. Both impose that the row weights of $G_2$ must be even, while in the former case, if we use the standard form $H_1=\left(\begin{array}{cc} A & \identity \end{array}\right)$, then the only way to create the all-ones row is to take a linear combination of every row, implying that the columns of $A$ have odd weight.

\subsection{Code Distance for CSS Codes and Majorana Fermion Error Correction}\label{sec:css}

As observed in \cite{kay2016c}, a distance 5 CSS code is, in principle, capable of correcting for any single phase error that occurs during the evolution of a perfect state transfer spin chain \cite{christandl2004,kay2010a}, simply requiring modification of how to act based on the error syndrome. This is mathematically described above, but conceptually, it boils down to the idea illustrated in Fig.\ \ref{fig:errors}: if we use an $X$-error correcting part of a CSS code first, we detect the positions of the $X$ or $Y$ operators in the decoding region. Since we know these represent the ends of Majorana fermions, these also tell us where there are sequences of $Z$ rotations. Once these have been corrected, the only remaining errors are a maximum of two $Z$ errors indicating whether the positions with $X$ errors were $X$ or $Y$s.

\begin{figure}
\begin{center}
\includegraphics[width=0.45\textwidth]{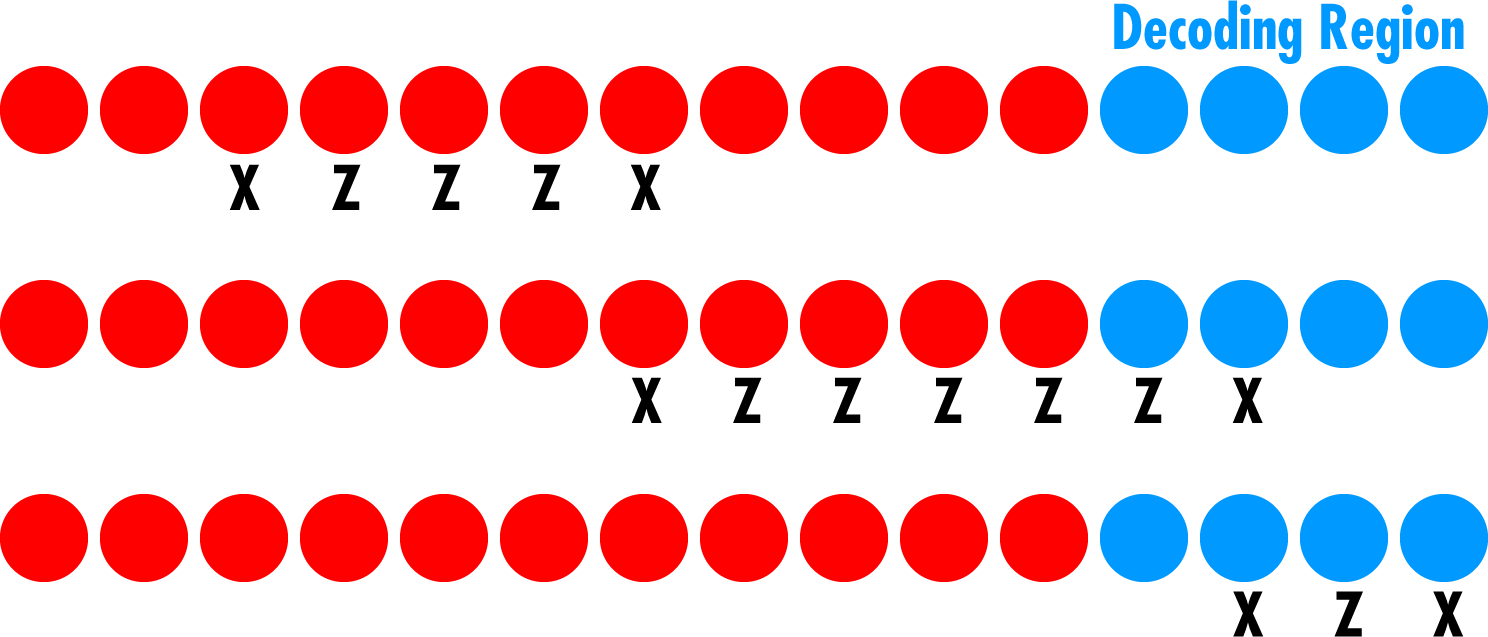}
\end{center}
\vspace{-0.5cm}
\caption{A 2-fermion error on a spin chain consists of a pair of $X$ or $Y$ rotations, with $Z$ rotations in-between. Detecting bit-flip locations in the decoding region implies the locations of intermediate phase errors to be corrected.}\label{fig:errors}
\vspace{-0.5cm}
\end{figure}

For the explicit construction given in \cite{kay2016c}, it proved sufficient to correct for 1 $Z$ error, i.e.\ forming the CSS code by combining a distance 5 and a distance 3 code instead of two distance 5 codes, permitting a smaller encoding region. We will now argue this is true for any CSS-based construction code using the $E'$ error association.
\begin{lemma}
CSS codes  with $d_1=5$, $d_2=3$ correct for a pair of Majorana fermion errors during state transfer.
\end{lemma}
\begin{proof}
Consider the parity check matrices $H_1$ and $G_2$ of code $C_1$ and the dual to $C_2$. $H_1$ detects the location of any pair of $X$ errors (for the Pauli code, or $c_n$ for the Majorana code; the distinction is unimportant here). If one or no $X$ errors were detected, the maximum number of $Z$ errors is 1, which can be corrected for using a distance 3 code. We then want to know that, if two $X$ errors occur, we can correct for any and all $Z$ errors that might arise. These $Z$ errors occur on the same sites as the $X$ errors. The only question is whether the syndrome for a pair of $Z$ errors occurring on a pair of qubits necessarily gives a non-trivial syndrome that is different from the syndromes for $Z$ errors on either of those qubits separately. Since $G_2$ is distance 3, all pairs of columns are linearly independent. Hence, distinct syndromes result. Interestingly, this means that we're using the non-degenerate code as if it were a degenerate code of greater distance.

\end{proof}

The construction of asymmetric quantum codes, including CSS codes, has already received some attention in the literature \cite{ezerman2011,evans2007,stephens2008,sarvepalli2009}. The codetables \cite{zotero-null-897} show that the smallest $M$ for which both $C_1$ and $C_2^\perp$ can exist is $M=13$. While this does not guarantee a solution satisfying $C_2\subseteq C_1$, \cite{sarvepalli2009} states its existence. We have found a suitable example, corresponding to case (i), and verified that there are no case (ii) examples (and hence no case (iii) examples either). The generators for the CSS code for $G_1$ and $G_2$ (above the line) are:
\begin{center}
\begin{tabular}{ccccccccccccc}
1 & 0 & 0 & 0 & 0 & 1 & 1 & 1 & 0 & 0 & 1 & 0 & 1 \\
 0 & 1 & 0 & 0 & 0 & 1 & 1 & 0 & 1 & 1 & 0 & 1 & 1 \\
 0 & 0 & 1 & 0 & 1 & 0 & 1 & 0 & 1 & 0 & 1 & 1 & 1 \\
 0 & 0 & 0 & 1 & 1 & 1 & 0 & 1 & 0 & 1 & 1 & 1 & 1 \\
 \hline
 0 & 0 & 0 & 0 & 1 & 0 & 0 & 1 & 1 & 1 & 0 & 0 & 1
\end{tabular}.
\end{center}

For case (iii), the smallest possible solution was found to have $M=15$. The generators take the form
\begin{center}
\begin{tabular}{ccccccccccccccc}
 1 & 0 & 0 & 0 & 1 & 1 & 0 & 0 & 1 & 0 & 1 & 1 & 1 & 0 & 1 \\
 0 & 1 & 0 & 0 & 0 & 1 & 1 & 1 & 0 & 0 & 1 & 1 & 0 & 1 & 1 \\
 0 & 0 & 1 & 0 & 1 & 0 & 1 & 1 & 1 & 1 & 1 & 1 & 0 & 0 & 0 \\
 0 & 0 & 0 & 1 & 0 & 0 & 0 & 1 & 1 & 1 & 0 & 1 & 1 & 1 & 1 \\
  \hline
0 & 0 & 0 & 0 & 1 & 1 & 0 & 1 & 0 & 0 & 0 & 0 & 1 & 1 & 1
\end{tabular}.
\end{center}

These describe the smallest possible CSS-type error correcting codes that correct for any arbitrary  two-Majorana fermion error during the evolution of the chain, including any single phase error, requiring either 13 or 15 qubits depending on the starting conditions. 
The case (iii) reduction is particularly dramatic compared to the 36 qubits required in \cite{kay2016c}, heralding a significant impact on the working regimes of a noisy spin chain.

\section{Restricted Noise Models}

In special cases, different assignments to the permutation $P$ can have benefits, facilitating a far more powerful application of CSS-based constructions. We will now restrict to perfect state transfer chains where the spectrum is symmetric about 0. This symmetry property is a feature of the standard perfect state transfer chain \cite{christandl2004}, and is used more broadly as it has the useful consequence that the chain has 0 magnetic field ($B_n=0$).
\begin{lemma}Define the odd-parity fermions to be $\{c_{2n-1}\}_{n=1}^{N/2}\cup\{c_{N+2n}\}_{n=1}^{N/2}$.
For a chain with a symmetric spectrum, time evolution preserves the parity of fermions.
\end{lemma}
\begin{proof}
Define the matrix $D=\sum_{n=1}^N(-1)^n\proj{n}$. We recognise that this anti-commutes with $h_1$. Hence, $Z\otimes D$ commutes with $Y\otimes h_1$. The eigenvalues of $Z\otimes D$ are therefore constants of the motion. Thus $c_1$, for example, can only evolve into a superposition of $c_{2n-1}$ or $c_{N+2n}$.
\end{proof}
A dephasing error $Z_n\equiv c_{n}c_{N+n}$ is an even and odd pair of fermions, and this is preserved by the time evolution. Identifying $X_n$ with even parity fermions and $Z_n$ with odd-parity fermions (i.e.\ using the choice $E$), we need only correct for up to one error of each parity. 

In comparison, the Hamiltonian $H$ is composed of pairs of fermions of the same parity. As such, this restricted noise model is not appropriate to timing errors (described as a power series expansion in $H$), or to perturbations in the coupling strengths. The existence of codes specialised to these usage scenarios, improving upon Sec.\ \ref{sec:css}, remains open.


\subsection{Perfect Quantum Codes}\label{sec:perfect}

Our target is thus to design a code that can correct for up to $r$ bit flip errors and up to $r$ phase flip errors, as Gottesman achieved in \cite{gottesman1997}. A CSS construction where both codes are distance $2r+1$ is sufficient because the two error types are corrected independently. What is remarkable is just how efficient the code becomes.

\begin{lemma}
A perfect classical code with parity-check matrix $H$ has a corresponding quantum code $H_1=G_2=H$, which is a perfect case (ii) solution for dephasing noise on spectral-symmetric chains.
\end{lemma}
\begin{proof}
The definition of a perfect code is that every state in the space is used for detecting a different error, i.e.\ for the classical codes, $2^M=2^k\sum_{i=0}^{\lfloor(d-1)/2\rfloor}\binom{M}{i}$, while for our quantum code, it would mean
$$
2^M=2^{2k-M}\left(\sum_{i=0}^{\lfloor(d-1)/2\rfloor}\binom{M}{i}\right)^2.
$$
The latter trivially follows from the former. The only non-trivial perfect codes are the Golay code and the Hamming codes. These are all weakly self-dual, meaning that $HH^T\equiv 0\text{ mod }2$, as required for the CSS construction. The row weights are all even, showing that all the stabilizers commute with $X^{\otimes M}$ and $Z^{\otimes M}$. Clearly these terms cannot be contained within the code because they mutually anti-commute -- they are logical operators of an encoded qubit. This means they are solutions for case (ii) (and consequently case (i)), but not case (iii).
\end{proof}
The Hamming (7,4,3) yields the usual Steane $[[7,1,3]]$ code \cite{steane1996a}. The other Hamming codes yield $[[2^r-1,2^r-2r-1,3]]$ codes that have capacity approaching 1, while correcting for one phase error during transfer. Meanwhile, the Golay code gives a $[[21,3,7]]$ code, whose additional distance hints at exciting prospects. 

Codes such as the 7 qubit Steane code are not usually considered to be perfect because the ability to correct for combinations of $X$ and $Z$ errors over and above the basic distance of the code is usually irrelevant, but it is absolutely essential in the present context.

While the best case (iii) solution is not a perfect code, the reduction to 10 qubits, with generators
\begin{center}
\begin{tabular}{cccccccccc}
 1 & 0 & 0 & 0 & 1 & 0 & 0 & 1 & 0 & 1 \\
 0 & 1 & 0 & 0 & 1 & 1 & 1 & 0 & 1 & 1 \\
 0 & 0 & 1 & 0 & 1 & 1 & 1 & 0 & 0 & 0 \\
 0 & 0 & 0 & 1 & 1 & 0 & 1 & 1 & 1 & 1 \\
 \hline
0 & 0 & 0 & 0 & 1 & 1 & 0 & 1 & 1 & 0
\end{tabular},
\end{center}
also represents a marked improvement.

\section{Simulations}

\begin{figure}[!tbp]
\begin{center}
\includegraphics[width=0.45\textwidth]{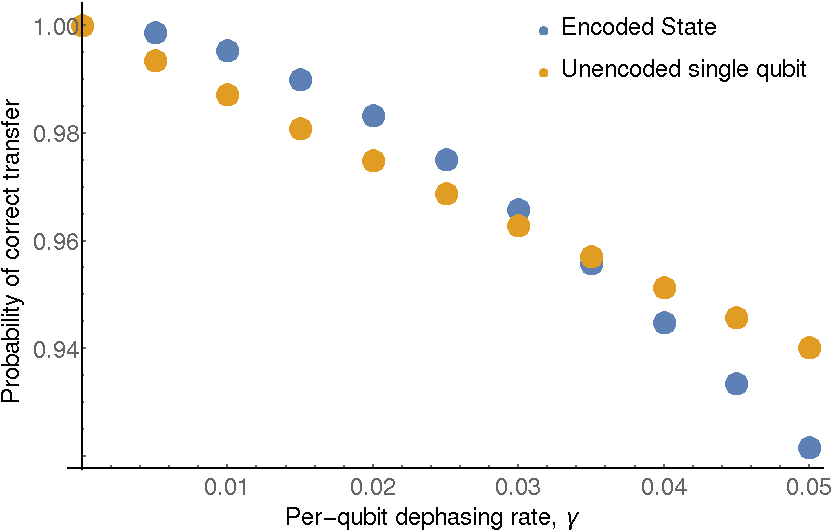}
\vspace{-0.2cm}
\caption{For a chain of length $N=12$, comparison of state transfer success for perfect encoding ($M=7$, averaged over all pure single qubit input states), and no encoding. The rest of the chain was initialised as $\ket{0}^{\otimes (N-M)}$.}\label{fig:EC_against_dephasing}
\end{center}
\vspace{-0.8cm}
\end{figure}

We have proven that the $[[7,1,3]]$ code can always correct for a single dephasing error, occurring at an unknown position and time during the transfer. However, it is more natural to consider dephasing noise, with dephasing rate $\gamma$, described by a Master equation
$$
\frac{d\rho}{dt}=-i[H,\rho]+\gamma\sum_{n=1}^NZ_n\rho Z_n-N\gamma\rho.
$$
We expect the code to be effective here as well -- for sufficiently small error rates ($2M\gamma t_0\sim 1$, recalling that $t_0\sim N$ if the maximum coupling strength of the chain is bounded), only one dephasing error is expected, which we know can be corrected. One can immediately see the critical importance of minimising $M$. The numerical simulation of Fig.\ \ref{fig:EC_against_dephasing} shows that there is an operating regime for small dephasing rates such that the encoding outperforms no encoding. To generate this, we selected the standard solution of perfect state transfer, $J_n=\sqrt{n(N-n)}$. Simulation of other similar chain lengths yielded indistinguishable curves, consistent with the predictions of \cite{kay2016c}.


Other scenarios in which a small number of Majorana fermions can be expected include Hamiltonian perturbations and timing errors. However, neither of these preserves the parity of the fermions involved, so our perfect codes cannot be used. The optimal codes of \ref{sec:css} (consisting of 13 or 15 qubits) yield qualitatively similar performance to those presented in \cite{kay2016c} (for a 15 qubit code).

\section{Impossibility of Correcting Bit-flip noise} \label{sec:impossible}

So far, we have concentrated on showing how to correct for errors that can be described as a small number of Majorana fermions. While it is not uncommon that one particular type of noise dominates another, a frequently studied noise model is depolarising noise, wherein errors of $X$, $Y$ and $Z$ types all act with equal likelihood. If we are to tolerate such noise, we must be able to correct for a single $X$ error occurring at any site, at any time. The purpose of this section is to argue that it is impossible to do this perfectly unless the encoding region is at least half the size of the whole spin chain, for which there are trivial solutions. We assume that the system outside the encoding and decoding regions is initialised to the all-zero state. To achieve this, we will specialise to showing that an error $X_{N/2+1}$ ($N$ even) at time $t_0/2$ cannot be perfectly corrected. This is expected to be the most destructive instance, requiring the largest number of fermions (by symmetry), maximising the proportion of a single transferring excitation on that site.

Rewriting the error $X_{N/2+1}$ as $Z_1Z_2\ldots Z_{N/2}c_{N/2+1}$, we can treat the evolution of the $c_{N/2+1}$ independently from the other terms, so let us neglect it for now. Instead, consider the time evolution of everything else:
$$
U=e^{-iH(2t_0-t_0/2)}Z_1Z_2\ldots Z_{N/2}e^{-iHt_0/2}
$$
We made the total evolution time $2t_0$ (perfect revival) instead of $t_0$ (state transfer) for simplicity; it does not affect the ultimate conclusion. We need a basis of states on the encoding region (of $M$ qubits), and while we could select the computational basis, there is a more natural basis to choose. Inspired by \cite{haselgrove2005}, we calculate the $M\times M$ matrix $W=\Pi_{in}U\Pi_{in}$, where $\Pi_{in}$ projects onto the first excitation subspace of the input spins. The eigenvectors of $W$, $\sum_{m=1}^M\lambda_{nm}\ket{m}$, define creation operators
$$
a_n^\dagger=\half\sum_{m=1}^M\lambda_{nm}(c_m-ic_{N+m}),
$$
which provide the requisite basis. Moreover, the corresponding eigenvalues $\lambda_n$ indicate the probability amplitude with which such an excitation returns to the encoding region at the revival time. These act independently, hence the probability amplitude that a pair of excitations both return is the product of the two eigenvalues.

Encoding a qubit requires two logical states, at least one of which must have some excitations. So, consider the component of that state which has the maximum number of excitations. Unless there is a value $|\lambda_n|=1$, there is a non-zero probability that none of these excitations arrive on the encoding/decoding region. Hence, there is a term $\proj{0}^{\otimes M}$ in the arrival density matrix (we chose the maximum excitation number to ensure that this component gives no off-diagonal terms). Since both logical states  have at least some component of $\proj{0}^{\otimes M}$, they cannot be perfectly distinguished. We conclude that the logical qubit cannot be perfectly error corrected.

We now reintroduce the evolution of $c_{N/2+1}$. Provided there is no value $|\lambda_n|=1$, the time evolution of $c_{N/2+1}(t_0/2)$ cannot localise onto just the encoding region -- that result would provide one such eigenvector. Considering the component that does not arrive on the encoding region, there continues to be a $\proj{0}^{\otimes M}$ element in the density matrix of each logical state.

The only assumption that remains is that there is no value $|\lambda_n|=1$. Let
$$
R=\Pi_{1}U\Pi_{1}=e^{ih_1t_0/2}\tilde D e^{-ih_1t_0/2},
$$
$\tilde D=\text{diag}(-1,-1,\ldots,-1,1,1,\ldots 1)$, and $\Pi_1$ is the projector onto the single excitation subspace for the whole chain. Given that $W$ is the $M^{th}$ principal sub-matrix of $R$, its eigenvalues are contained within the range of eigenvalues of $R$, which are $\pm 1$ given that $R^2=\identity$. For Hermitian matrix $W$, with entries $w_{ij}$, the eigenvalues are contained within the discs
$$
w_{ii}\pm\sqrt{\sum_{j\neq i}|w_{ij}|^2}.
$$
However, Lemma \ref{lem:0diagonal} will prove that $w_{ii}=0$. Moreover, if we took the entire row of $R$, then the sum-mod-square is 1. Since Lemma \ref{lem:antidiag} proves that $|w_{i,N+1-i}|>0$, and hence
$$
\sum_{j=1}^{N/2}|w_{ij}|^2<1,
$$
for $i\leq N/2$, $W$ does not have any singular values equal to 1 if the size of the encoding region is $M\leq N/2$.

\begin{lemma} \label{lem:0diagonal}
Assuming $h_1$ has a symmetric spectrum, the elements $\bra{n}R\ket{n}=0$ for all $n$.
\end{lemma}
\begin{proof}
If we let $e^{-ih_1t_0/2}\ket{n}=\sum_m\alpha_m\ket{m}$, then it is sufficient to show that $|\alpha_m|=|\alpha_{N+1-m}|$ for all $m$. 

By symmetry, we have that
$$
e^{-ih_1t_0/2}\ket{N+1-n}=\sum_m\alpha_{N+1-m}\ket{m}.
$$
Since $Dh_1D=-h_1$ provided $h_1$ has a symmetric spectrum, up to a $\pm 1$ phase,
$$
e^{ih_1t_0/2}\ket{N+1-n}=D\sum_m\alpha_{N+1-m}\ket{m}.
$$
However, by the perfect state transfer property, we have $e^{ih_1t_0/2}\ket{N+1-n}=e^{-ih_1t_0/2}\ket{n}$, up to a global phase. Thus, up to a global phase, we have
$$
\sum_m\alpha_m\ket{m}=D\sum_m\alpha_{N+1-m}\ket{m},
$$
as required.
\end{proof}

\begin{lemma}\label{lem:antidiag}
Assuming $H$ is the standard perfect state transfer chain of \cite{christandl2004}, and $N$ even, the elements $\bra{n}R\ket{N+1-n}\neq0$ for all $n$.
\end{lemma}
We believe this to be true more generally, but the extension remains unproven.
\begin{proof}
Given that 
\begin{multline*}
e^{ih_1t_0/2}=-\frac{1}{2^{(N-1)/2}}\sum_{n,m=0}^{N-1}\ket{n+1}\bra{m+1}\times\\\times\sqrt{\binom{N-1}{n}\binom{N-1}{m}}i^{n+m}\ _2F_1(-n,-m;1-N;2)
\end{multline*}
for the standard perfect state transfer chain, we can compare this to the eigenvectors of the system, explicitly given in \cite{albanese2004}. This reveals that $e^{-ih_1t_0/2}\ket{m}$ and $\ket{\lambda_m}$ have amplitudes on each site of equal magnitude. Hence,
$$
\bra{n}R\ket{m}=\bra{\lambda_n}\tilde D\ket{\lambda_m}
$$
Since the diagonal matrix flips the symmetry of a vector, the elements $\bra{n}R\ket{m}=0$ if $n,m$ have the same parity, subsuming Lemma \ref{lem:0diagonal}. Furthermore, $R$ is Hermitian (and imaginary). Concentrating on $m=N+1-n$
$$
\bra{\lambda_n}\tilde D\ket{\lambda_{N+1-n}}=-2\sum_{k=1}^{N/2}\lambda_{n,k}^2(-1)^{k+1},
$$
since we know that the eigenvector elements satisfy $\lambda_{N+1-n,k}=(-1)^{k+1}\lambda_{n,k}$. Substituting the eigenvector relation $\lambda_n\lambda_{n,k}=J_{k-1}\lambda_{n,k-1}+J_k\lambda_{n,k+1}$ causes neighbouring terms in the sum to cancel, leaving only 
$$
\bra{\lambda_n}\tilde D\ket{\lambda_{N+1-n}}=2\lambda_{n,N/2+1}^2\frac{J_{N/2}}{\lambda_n}(-1)^{N/2+1}.
$$
For even $N$, the eigenvalues cannot be 0, as eigenvalues occur in $\pm\lambda$ pairs and are non-degenerate. Furthermore, since $|\lambda_{n,N/2+1}|=|\lambda_{n,N/2}|$, if it were the case that $\lambda_{n,N/2+1}=0$, there would be two consecutive 0s, and the only possible solution in that case is for the entire eigenvector to be 0. We therefore conclude that $\bra{\lambda_n}\tilde D\ket{\lambda_{N+1-n}}\neq 0$ for all $n$.
\end{proof}

In contrast to this result that error correction of a single bit flip is impossible unless $M>N/2$, \cite{marletto2012} explicitly shows, once $M>N/2$, how to find states that at time $t_0/2$ are entirely localised on the first half of the spin chain, and are hence entirely unaffected by the act error operation. Indeed, with $M=N/2+1$, at least 2 $a_n^\dagger$s return perfectly, call them $a_1^\dagger$ to $a_2^\dagger$. Define the two logical states to be $\ket{0}^{\otimes N}$ and $\ket{1}^{\otimes M}\ket{0}^{\otimes(N-M)}$. After the evolution, the encoding/decoding region of the first logical state never contains more than 1 excitation (coming from $c_{N/2+1}$), while the second always has at least 2 excitations (as the original encoding certainly contained $a_1^\dagger$ and $a_2^\dagger$, neither of which can be removed by $c_{N/2+1}$), so this guarantees that this specific error can be corrected, and indeed any single $X$ error at any position and time because while the appropriate $a_1^\dagger$ and $a_2^\dagger$ are different for every position and time, no matter the basis, the state $\ket{1}^{\otimes M}$ always contains them. An alternative perspective is that if $M=N/2+1$, then even if the maximum number of excitations ($N/2-1$) are left outside the encoding/decoding region, there must be at least 2 remaining on the encoding/decoding region.

We have therefore exactly determined when error correction of a single unknown bit-flip error is possible, at least for the standard instance of perfect state transfer. However, it requires such a large encoding region that trivial solutions arise (encode on a single common spin, and decode from that spin immediately), and we consider it outside the bounds of interest for spin chains, for which the encoding/decoding regions should be much shorter than the length of the chain.

\section{Conclusions}

In this paper, we have studied the error correction of perfect state transfer, improving markedly upon the results of \cite{kay2016c} for dephasing noise by specifying optimal error correcting strategies. It is important to note that the method of error correction is largely independent of the spin chain; we used the fact that the chain was capable of perfect state transfer, and the best error correcting codes required a perfect state transfer chain with no magnetic field. One typically expects that CSS codes are not the best, and that they can be out-performed by error correcting codes that do not separate out the $X$ and $Z$ errors into distinct cases. The perfect 5-qubit code \cite{laflamme1996} is one such example of an improvement over the best, 7 qubit, CSS code \cite{steane1996a}. There is also an 11 qubit, distance 5, non-CSS code \cite{gottesman1997}. However, in the present context, we used CSS codes far more efficiently than in the usual Pauli-error context and, indeed, the fact that the resultant codes are perfect for Majorana errors proves that there can be no better.

This study was predicated on the assumption of choosing a perfect state transfer system. However, as soon as there is noise present, transfer is not perfect. At that point, we should assess whether there are imperfect transfer schemes that are more tolerant of noise. Indeed, there are. Solutions such as \cite{apollaro2012} provide high fidelity transfer at shorter times. Since the time is shorter, the dephasing has less effect, and this difference can more than compensate for the imperfect transfer fidelity in the non-noisy case. The crucial step for future studies, in terms of improving the operating regime depicted in Fig.\ \ref{fig:EC_against_dephasing}, is to optimise the chains used. That said, this will require a non-trivial step in the theory because solutions such as those of \cite{apollaro2012} are unsuitable without further modification. Those schemes are tuned specifically for end-to-end transfer. They generate high transfer fidelity between $\ket{1}\rightarrow\ket{N}$ at a higher speed, at the cost of the transfer fidelity between intermediate sites. However, error correction requires high quality transfer for all pairs $\ket{n}\rightarrow\ket{N+1-n}$ where $n\leq 7$ is in the encoding region. For example, the optimal $N=42$ solution from \cite{apollaro2012} has a $1\rightarrow 42$ excitation transfer fidelity of 0.993 in a time $0.77t_0$ (where $t_0$ is the perfect transfer time for a system with the same maximum coupling strength), but the fidelity for $5\rightarrow 38$ is less than 0.4. 

While error correction of phase errors is possible, we have shown that exact error correction of other errors, such as a single bit-flip, is not possible during a perfect state transfer. We are thus constrained to using systems for which phase noise is dominant over depolarising noise, and for which state transfer can be achieved in a time shorter than the depolarising time. Importantly, this proof is not contingent upon the description of the error correction in terms of Majorana fermions -- there is no description that permits error correction of these bit flip errors. That said, while perfect correction of even a single bit flip is impossible, the effects can be mitigated via an approximate error correction. A repetition code followed by majority vote can be expected to have some advantage. For example, on a chain of 21 qubits, and a single $X$ error occurring at $t_0/2$ on qubit 11, an unencoded transfer has an error probability of 85\%, while a 3-qubit repetition code reduced that to 7\%, while a 5-qubit code yields a further improvement to 0.0006\%. Clearly, this has the potential to be extremely effective against bit-flip errors, and could be combined with the codes developed here for phase flips in order to provide reasonable robustness against depolarising noise.

Some aspects of the formalism, although not the setting, are similar to that developed in \cite{vijay2017}, discovered as this work was being finalised.

%

\end{document}